%% file: main.tex
\documentclass[11pt]{article}

\usepackage{amsmath, amssymb, amsthm, amsfonts,bm,mathtools}
\usepackage[inline]{enumitem}
\usepackage{csquotes}
\usepackage{hhline}
\usepackage{soul}
\usepackage{cite}
\usepackage{framed}
\usepackage[framemethod=tikz]{mdframed}
\usepackage{appendix}
\usepackage{graphicx}
\usepackage{color}
\usepackage{wrapfig}
\usepackage{multirow}
\usepackage{algorithm, subcaption}
\usepackage[noend]{algpseudocode}
\usepackage[textsize=tiny]{todonotes}
\usepackage{enumitem}
\usepackage{amssymb}
\setitemize{noitemsep,topsep=3pt,parsep=3pt,partopsep=3pt}

\usepackage[margin = 1in]{geometry}

\usepackage{xspace}

\definecolor{darkgreen}{rgb}{0,0.5,0}
\definecolor{darkblue}{rgb}{0,0,0.8}
\usepackage{hyperref}
\hypersetup{
    unicode=false,          
    colorlinks=true,        
    linkcolor=darkblue,          
    citecolor=darkgreen,        
    filecolor=magenta,      
    urlcolor=cyan           
}
\RequirePackage[]{silence}
\usepackage[nameinlink,capitalize]{cleveref}

\theoremstyle{theorem}
\newtheorem{theorem}{Theorem}[section]
\theoremstyle{lemma}
\newtheorem{lemma}[theorem]{Lemma}
\theoremstyle{corollary}
\newtheorem{corollary}[theorem]{Corollary}
\theoremstyle{claim}

\theoremstyle{definition}

\theoremstyle{remark}

\newcommand{\diameter}[1]{\ensuremath{\textup{diameter}{\left(#1 \right)}}\xspace}
\newcommand{\diameterG}{\diameter{G}}

\newcommand{\DenseSubgraph}[1]{\textsc{DenseSubgraph}(#1)}
\newcommand{\DensestSubgraph}[1]{\textsc{DenseSubgraph}(#1)}

\newcommand{\DUAL}[1]{\textsc{Dual}(#1)}
\newcommand{\PRIMAL}[1]{\textsc{Primal}(#1)}
\newcommand{\Dual}{\textsc{Dual}}
\newcommand{\Primal}{\textsc{Primal}}

\newcommand{\ignore}[1]{}

\algnewcommand\algorithmicswitch{\textbf{switch}}
\algnewcommand\algorithmiccase{\textbf{case}}

\algdef{SE}[SWITCH]{Switch}{EndSwitch}[1]{\algorithmicswitch\ #1\ \algorithmicdo}{\algorithmicend\ \algorithmicswitch}%
\algdef{SE}[CASE]{Case}{EndCase}[1]{\algorithmiccase\ #1}{\algorithmicend\ \algorithmiccase}%
\algtext*{EndSwitch}%
\algtext*{EndCase}%

\newcommand{\eps}{\epsilon}

\newcommand{\congest}{\ensuremath{\mathsf{CONGEST}~}}
\newcommand{\local}{$\mathsf{LOCAL}$\xspace}

\newcommand{\poly}{\operatorname{\text{{\rm poly}}}}

\newcommand{\floor}[1]{\lfloor #1 \rfloor}

\newcommand{\dist}{\operatorname{dist}}
\newcommand{\outdeg}{\operatorname{outdeg}}
\newcommand{\indeg}{\operatorname{indeg}}

\renewcommand{\paragraph}[1]{\vspace{0.15cm}\noindent {\bf #1}:}

\begin{document}

\title{Distributed  Dense Subgraph Detection and Low Outdegree Orientation}
\author{ Hsin-Hao Su \\ Boston College \\ suhx@bc.edu \and Hoa T. Vu \\ San Diego State University \\ hvu2@sdsu.edu }
\date{}

\maketitle\begin{abstract}
The densest subgraph problem, introduced in the 80s by Picard and Queyranne [Networks 1982] as well as Goldberg [Tech. Report 1984], is a classic problem in combinatorial optimization with a wide range of applications.  The lowest outdegree orientation problem is known to be its dual problem. We study both the problem of finding dense subgraphs and the problem of computing a low outdegree orientation in the distributed settings.

Suppose $G=(V,E)$ is the underlying network as well as the input graph. Let $D$ denote the density of the maximum density subgraph of $G$. Our main results are as follows.
\thispagestyle{empty}
\begin{itemize}
\item Given a value $\tilde{D} \leq D$ and $0 < \epsilon < 1$, we show that a subgraph with density at least $(1-\epsilon)\tilde{D}$ can be identified deterministically in $O((\log n) / \epsilon)$ rounds in the \textsf{LOCAL} model. We also present a lower bound showing that our result for the \textsf{LOCAL} model is tight up to an $O(\log n)$ factor. 

In the \textsf{CONGEST} model, we show that such a subgraph can be identified in $O((\log^3 n) / \epsilon^3)$ rounds with high probability. Our techniques also lead to an $O(\diameterG + (\log^4 n)/\epsilon^4)$-round algorithm that yields a $1-\epsilon$ approximation to the densest subgraph. This improves upon the previous $O(\diameterG /\epsilon  \cdot  \log n)$-round algorithm by Das Sarma et al. [DISC 2012] that only yields a $1/2-\epsilon$ approximation.

\item Given an integer $\tilde{D} \geq D$ and $\Omega(1/\tilde{D}) < \epsilon < 1/4$, we give a deterministic, $\tilde{O}((\log^2 n) /\epsilon^2)$-round \footnote{$\tilde{O}$ hides $\log \log n$ factors.}  algorithm in the \congest model that computes an orientation where the outdegree of every vertex is upper bounded by $(1+\epsilon)\tilde{D}$. Previously, the best deterministic algorithm and randomized algorithm by Harris [FOCS 2019] run in $\tilde{O}((\log^6 n)/ \epsilon^4)$ rounds and $\tilde{O}((\log^3 n) /\epsilon^3)$ rounds respectively   and only work in the \local model.
\end{itemize}
\end{abstract}\newpage
\setcounter{page}{1}
\input{intro}

\input{local}
\input{congest}

\input{CONGEST_splitting}
\section*{Acknowledgment} We thank David Harris for pointing out Lemma \ref{lem:bits}, which can be used to improve the running time of the dual rounding algorithm in our previous version. We thank Quanquan Liu for pointing out a subtle bug and suggesting a fix for the proof of  Lemma \ref{lem:dual_feasible} in the previous version.
\bibliographystyle{alpha}
\bibliography{references}
\appendix
\input{appendix}
\end{document}

%% file: intro.tex
\sloppy
\section{Introduction}

\paragraph{The Dense Subgraph Problem} Given a graph $G = (V,E)$, the {\it maximum density subgraph} problem (or the densest subgraph problem) is to find a subgraph $H$, where its density $d(H) = |E(H)|/|V(H)|$ is maximized over all subgraphs of $G$. We denote the density of the maximum density subgraph of $G$, $\max_{H \subseteq G} d(H)$, by $D$.

First studied by Picard and Queyranne \cite{PQ82} as well as Goldberg \cite{Goldberg84}, the maximum density subgraph problem has found numerous applications in community detection in social networks \cite{DGP07, CS12}, link spam identification \cite{GKT05, BXGPF13}, and computational biology \cite{FNBB06, SHKRZ10}. Faster algorithms \cite{GGT89, Charikar00, KS09} have been developed for the problem and its variants since then. Moreover, the problem has been widely studied under different models of computation such as the streaming settings \cite{BHNT15, EsfandiariHW15, McGregorTVV15}, the dynamic setting \cite{SJ20,BHNT15}, the massive parallel computation settings \cite{BahmaniKV12, BahmaniGM14, GLM19}, and the distributed settings \cite{SarmaLNT12}.  

We study the problem of detecting dense subgraphs in the distributed settings, namely, in both the \local and the \congest models. Let $n$ and $m$ be the number of vertices and the number of edges respectively. Furthermore, let $\Delta$ be the maximum degree. In such models, vertices are labeled with unique IDs and they operate in synchronized rounds. In each round, each vertex sends a message to each of its neighbors, receives messages from its neighbors, and performs local computations. The time complexity of an algorithm is defined to be the number of rounds used. In the \local model, the message size can be arbitrary. In the \congest model, the message size is upper bounded by $O(\log n)$ bits. We consider the following parameterized version of the maximum density subgraph problem in the distributed settings, which may capture the computational nature of some applications such as how a dense community can be found by only communicating with the neighbors in social networks.

\begin{mdframed}[hidealllines=false,backgroundcolor=white]
$\DensestSubgraph{\tilde{D}, \epsilon}$ : Given a graph $G = (V,E)$, a parameter $\tilde{D} \geq 0$, and $0 < \epsilon < 1$, every vertex $u$ outputs a value $h_u \in \{0,1\}$ such that $d(H) \geq (1-\epsilon) \tilde{D}$ where $H = \{u \mid h_u = 1\}$. If $\tilde{D} \leq D$ then $H$ must be non-empty. 
\end{mdframed}

The first question we investigate is whether $\DensestSubgraph{\tilde{D}, \epsilon}$ can be solved locally.  Intuitively, most dense subgraphs have short diameters because they are well-connected.   Thus, they  can be detected locally. Our first result justifies this intuition.

\begin{theorem}\label{thm:main-local-1} There exists a deterministic algorithm for $\DensestSubgraph{\tilde{D}, \epsilon}$ that runs in $O((\log n) / \epsilon)$ rounds in the \local model.\end{theorem}

We will also present a lower bound showing that the running time of the algorithm is tight up to an $O(\log n)$ factor. The algorithm for the \local model uses large message size. This begs the question of whether the problem can be solved in the \congest model while remaining in the $\poly(1/\epsilon, \log n)$-round regime. We show that this is indeed possible with randomization. 

\begin{theorem}\label{thm:main-congest-1}
There exists a randomized algorithm that solves $\DensestSubgraph{\tilde{D}, \epsilon}$ w.h.p.\footnote{W.h.p.~denotes with high probability, which means with probability at least $1 - 1/n^c$ for an arbitrarily large constant $c$.},~and runs in $O((\log^3 n) / \epsilon^3)$ rounds in the  \congest model.
\end{theorem}

Finding the densest subgraph in such distributed settings  inevitably requires $\Omega(\diameter{G})$ rounds (e.g., consider two subgraphs of different densities connected by a path of length $\Omega(\diameter{G})$). Das Sarma et al.~\cite{SarmaLNT12} gave an algorithm for finding a $(1/2-\epsilon)$ approximation to the densest subgraph in $O(\diameter{G}/\eps \cdot \log n)$ rounds in the \congest model. We show that the approximation factor can be improved and the dependency on $\diameter{G}$ can be made additive:
\begin{corollary}\label{cor:main-congest-2}
There exists a randomized algorithm that finds a $(1-\epsilon)$-approximation to the maximum density subgraph w.h.p.~and runs in $O(\diameter{G} +(\log^4 n)/ \epsilon^4 )$ rounds in the \congest model.
\end{corollary}

Inspired by web-graphs, Kannan and Vinay \cite{KV99} defined the notion of density in directed graphs. Suppose that $G = (V,E)$ is a directed graph. The density of a pair of sets $S,T \subseteq V$ is defined as $d(S,T) = \frac{|E(S,T)|}{\sqrt{|V(S)||V(T)|}}$, where $E(S,T)$ denote the set of edges that go from a vertex in $S$ to a vertex in $T$. Note that we are assuming messages can go in both directions of an edge. Our result in the \local model can be generalized to the directed version of the problem. 

\paragraph{The Low Outdegree Orientation Problem} Given an undirected graph $G=(V,E)$, an $\alpha$-orientation is an orientation of the edges such that the outdegree of every vertex is upper bounded by $\alpha$. Picard and Queyranne \cite{PQ82} observed that an $\alpha$-orientation exists if and only if $\alpha \geq \lceil D \rceil$. Charikar \cite{Charikar00} formulated the linear program (LP) for the densest subgraph problem, and its dual is the fractional version of the lowest outdegree orientation problem \cite{BahmaniGM14}. 

An $\alpha$-orientation can be used to obtain a decomposition of the graph into $\alpha$ {\it pseudoforests} \cite{PQ82}, where a pseudoforest is a collection of disjoint {\it pseudotrees} (a pseudotree is a connected graph containing at most one cycle). The relation between the pseudoforest decomposition and the maximum density is analogous to that of the forest decomposition and the arboricity shown by Nash-Williams \cite{NW64}.

We consider the low outdegree orientation problem in the distributed setting. The problem can be formally stated as follows.
\begin{mdframed}[hidealllines=false,backgroundcolor=white]
Given a graph $G = (V,E)$, an integer parameter $\tilde{D} \geq D$, and $0 < \epsilon < 1$, compute a $(1+\epsilon)\tilde{D}$-orientation. The orientation of an edge is decided by either  of its endpoints.
\end{mdframed}
Our contribution for this problem is a deterministic, $\tilde{O}((\log^2 n )/ \epsilon^2)$-round algorithm in the \congest model.  Previously, the best deterministic algorithm and randomized algorithm by Harris \cite{Harris19} run in $\tilde{O}((\log^6 n) / \epsilon^4)$ rounds and $\tilde{O}((\log^3 n) /\epsilon^3)$ rounds respectively   and only work in the \local model.

\begin{theorem}\label{thm:lowdeg}Given an integer $\tilde{D} \geq D$, for any $32/\tilde{D} \leq \epsilon \leq 1/4$, there exists a deterministic algorithm in the \congest model that computes a $(1+\epsilon)\tilde{D}$-orientation and runs in
$$O\left(\frac{\log n}{\epsilon^2} + \left(\min(\log \log n, \log \Delta)+  \log(1/ \epsilon)\right)^{2.71}  \cdot (1/\epsilon)^{1.71}\cdot \log^2 n \right) \leq  \tilde{O}((\log^2 n) /\epsilon^2)  \mbox{ rounds}.$$
\end{theorem}     
\begin{table}[h] \centering
\caption{Previous results on the low outdegree orientation problem in the distributed setting.}\label{tbl:lowdeg}
\begin{tabular}{|l|l|l|l|l|}
\hline
Reference                  & Time & Model & Approx. & Rand.~or Det.\\ \hline
Barenboim and Elkin     \cite{BE10}   &    $O((\log n) / \epsilon)$   & \textsf{CONGEST}      &    $2+\epsilon$           &              Det. \\ \hline
Ghaffari and Su       \cite{GS17}      &  $O((\log^4 n) / \epsilon^3)$    &   \textsf{LOCAL}    &      $1+\epsilon$         &     Rand.          \\ \hline
Fischer et al.~\cite{FGK17} &   $ 2^{O(\log^2 (1/\eps \cdot \log n ))} $ &   \textsf{LOCAL}    &       $1+\epsilon$        &   Det.            \\ \hline
Ghaffari et al.~\cite{GHK18}  & $O((\log^{10} n  \cdot \log^{5} \Delta )/ \epsilon^9)$     &   \textsf{LOCAL}    &     $1+\epsilon$          &   Det.            \\ \hline
Harris   \cite{Harris19}                   & $\tilde{O}((\log^6 n) / \epsilon^4)$     &    \textsf{LOCAL}   &     $1+\epsilon$                   & Det.          \\ \hline
Harris       \cite{Harris19}               & $\tilde{O}((\log^3 n) / \epsilon^3)$       &   \textsf{LOCAL}     &     $1+\epsilon$          & Rand.         \\ \hline
{\bf new}                         &  $\tilde{O}((\log^2 n)/\epsilon^2)$    &   \textsf{CONGEST}    &     $1+\epsilon$          &   Det.            \\ \hline
\end{tabular}
\end{table}
 Table \ref{tbl:lowdeg} summarizes previous algorithms for this problem in the \local and \congest models. It is worthwhile to note that $\lceil D \rceil \leq a(G) \leq \lceil D \rceil + 1$ \cite{PQ82}, where $a(G)$ is the arborcity of the graph, as several previous results were parameterized in terms of $a(G)$.  
 
 Barenboim and Elkin \cite{BE10} introduced the H-partition algorithm that obtains a $(2+\epsilon)a(G)$-orientation as well as a $(2+\epsilon)a(G)$ forest decomposition. Ghaffari and Su \cite{GS17} observed that the problem of computing  a $(1+\epsilon)a(G)$-orientation reduces to computing of maximal independent sets (MIS) on the conflict graphs formed by augmenting paths. The MIS can be computed efficiently by simulating Luby's randomized MIS on the conflict graph. Fischer et al.~\cite{FGK17} initiated the study of computing such MIS (i.e.~the maximal matching in hypergraphs) deterministically. They gave a deterministic quasi-polynomial (in $r$) algorithm for computing the maximal matching in rank $r$ hypergraphs, resulting in a $2^{O(\log^2 (1/\eps \cdot \log n))}$ round algorithm for the orientation problem. Later, the dependency on $r$ has been improved to polynomial by \cite{GHK18}, which results in a $O((\log^{10} n  \cdot \log^5 \Delta) /\epsilon^9)$ rounds algorithm for the orientation problem. Recently, the deterministic running time for the problem is further improved by Harris \cite{Harris19} to $\tilde{O}((\log^6 n) / \epsilon^4)$ as they developed a faster algorithm for the hypergraph maximal matching problem. It is unclear if the above approaches via maximal matching in hypergraphs can be implemented in the \congest model without significantly increase on the number of rounds.  The low outdegree orientation problem has also been  studied in the centralized context by \cite{GW92,BF99, Kowalik06,GKS14, KKPS14,BB20}.

\paragraph{Our methods and contributions in a nutshell} Our first contribution is a simple yet powerful observation that dense subgraphs have low-diameter approximations.  We first give a simple proof via low-diameter decomposition \cite{LinialS93,MillerPX13}  that there exists a subgraph with diameter $O((\log n)/\epsilon)$ that has density at least $(1-\epsilon)D$ (Lemma \ref{lem:low-diameter-ds}). Hence,  if each vertex examines its local neighborhood up to a small radius, at least one vertex gets a good estimate of the densest subgraph. With appropriate bookkeeping, this leads to our {\em deterministic} $O((\log n)/\epsilon)$-round algorithm for detecting dense subgraphs in the \local model. We complement this algorithm with a lower bound showing that $\Omega(1/\epsilon)$ rounds are necessary (Lemma \ref{lem:local-lowerbound}).

In the \congest model, the starting point for both problems of detecting dense subgraphs and low outdegree orientation is the adaptation of the multiplicative weights update algorithm of Bahmani et al.~\cite{BahmaniGM14}. Their algorithm solves the dual linear program for $\DensestSubgraph{\tilde{D}, \epsilon}$ (which is the linear program for low outdegree orientation) using the multiplicative weights method. They also showed how to round the dual program's fractional solution to find a dense subgraph. While it appears that their algorithm can be implemented directly in the \congest model, there are a few issues that we need to resolve.

A naive implementation of the algorithm in \cite{BahmaniGM14} to solve the dual linear program uses $O(\diameterG)$ rounds per iteration. Furthermore, for the dense subgraph detection problem $\DensestSubgraph{\tilde{D}, \epsilon}$, the rounding procedures in \cite{Charikar00, BahmaniGM14} are inherently global. This is because they require checking whether certain subgraphs have high enough density. These subgraphs may have large diameters. Our contribution here is to remove the dependence on $\diameterG$ using ideas from Lemma \ref{lem:low-diameter-ds} along with appropriate bookkeeping. This results in a randomized algorithm that runs in  $O((\log^3 n) / \epsilon^3)$ rounds instead of $O(\diameterG \cdot \poly(\log n,1/\epsilon))$ rounds. Our adaption also removes the explicit use of real-valued weights and only use integers, which can be transmitted easily in the \congest model.

For the low outdegree orientation problem, we still need to develop a procedure to round the dual fractional solution obtained from the multiplicative weights update method into an integral solution efficiently and deterministically in the \congest model. Our contribution here is twofold. We use the idea of recent rounding-type algorithms for distributed matching \cite{Fischer17, AKO18}, where we process the fractional solution bit-by-bit and round the solution at each bit scale. We show that each scale reduces to the directed splitting problem  where there are known algorithms for the \local model \cite{GS17,GHKMSU17}. We will show how to modify these directed splitting algorithms to run in the \congest model. We provide an analysis showing that the total error incurred by multiple bit scales is small so that each outdegree is at most $(1+\eps)D$ at the end of the rounding procedure.

\paragraph{Organization} In Section \ref{sec:local}, we present our deterministic algorithm for the dense subgraph problem in the \local model. Section \ref{sec:congest} presents a randomized algorithm for the dense subgraph problem in the \congest model. Section \ref{sec:loo} exhibits a deterministic algorithm for the low outdegree orientation problem in the \congest model. Some proofs are deferred to the Appendix.

%% file: local.tex

\section{Deterministic Dense Subgraph Detection in the LOCAL Model}\label{sec:local}

In this section, we investigate the locality of the densest subgraph problem. In particular, we show how to solve $\DenseSubgraph{\tilde{D},\eps}$ \emph{deterministically} in the \local model. Combining with ideas in this section, we show that this problem can also be solved in the \congest model with randomization in the next section. 

\paragraph{The locality of the densest subgraph} We present an algorithm to solve $\DenseSubgraph{\tilde{D},\eps}$ in $O((\log n)/\epsilon)$ rounds deterministically in the \local model. We first give a structural lemma showing that for some sufficiently large constant $K$, there exists a subgraph with diameter at most $K/\eps \cdot \log n$ that has density at least $(1-\eps)D$.

\begin{lemma} [Densest subgraph's locality] \label{lem:low-diameter-ds}
 For all simple graphs, there exists a subgraph with diameter at most $K/\eps \cdot \log n$ for some sufficiently large constant $K$ that has density at least $(1-\eps)D$.
\end{lemma}
\begin{proof}
 It can be shown that  for any simple graph $G$ with $n$ vertices and $m$ edges, we can decompose $G$ into disjoint components such that each component has diameter at most $K/\eps \cdot \log n$ for some sufficiently large constant $K$ and furthermore the number of inter-component edges is at most $\eps m$ \cite{LinialS93, MillerPX13,Awerbuch85} (see also Theorem \ref{thm:low-diameter-decomposition}). This is known as the {\em low-diameter decomposition}.

Consider the densest subgraph ${H^*} \subseteq G$ with $n^*$ vertices and $m^*$ edges.  We apply the low-diameter decomposition to ${H^*} $ and let the components be $H_1^*,\ldots,H_t^*$. Let $n_i^*$ and $m_i^*$ be the number of vertices and edges in $H_i^*$ respectively. Suppose that $d(H_i^*) < (1-\eps)D$ for all $i$. Then, $m_i^*/n_i^* < (1-\eps) D$ which implies $m_i^* < (1-\eps) m^* n_i^*/n^*$. Thus,
\begin{align*}
\sum_{i = 1}^t m_i^* <  (1-\eps) \sum_{i=1}^t \frac{ m^* n_i^*}{n^*} = (1-\eps) m^*~.
\end{align*}
This implies that the number of inter-component edges is more than $\eps m^*$ which is a contradiction. Therefore, at least one component $H_j^*$ must have density  $d(H_j^*) \geq (1-\eps)D$.
\end{proof}

Using Lemma \ref{lem:low-diameter-ds}, we can design a \local algorithm to solve $\DenseSubgraph{\tilde{D},\eps}$ in $O( (\log n)/\epsilon)$ rounds. See Algorithm \ref{alg:local-1}.

 \begin{algorithm}
\caption{Distributed Algorithm for $\DenseSubgraph{\tilde{D},\eps}$ in the \local model}\label{alg:local-1}
\begin{algorithmic}[1]\small
\State Initialize $h_v \leftarrow 0$ for all vertices $v$.  
\State Let $r = K/\eps \cdot \log n$ for some sufficiently large constant $K$.
\State  Each vertex $v$ collects the subgraph induced by its $r$-neighborhood   $N_r(v)=\{ u: \dist(v,u) \leq r \}$ and  computes the densest subgraph $H(v)$ in $G[N_r(v)]$, i.e., the subgraph induced by $N_r(v)$. \label{ln:local-collect}
\If{$d(H(v)) \geq (1-\eps)\tilde{D}$}
\State $v$ becomes an active vertex and collects its $2r$-neighborhood  $N_{2r}(v)$. \label{ln:local-collect-2}
\EndIf

\If{an active vertex $v$ has the smallest ID among $N_{2r}(v)$}
\State $v$ becomes a black vertex.
\EndIf
\State Each black vertex $v$ broadcasts $H(v)$ to $N_{r}(v)$ and set $h_u \leftarrow 1$ \text{ for all $u$ in $V(H(v))$}. \label{ln:local-collect-3}
\end{algorithmic}
\end{algorithm}
\begingroup
\begin{NoHyper}
\def\thetheorem{\ref{thm:main-local-1}}
\begin{theorem}
There exists a deterministic algorithm for $\DensestSubgraph{\tilde{D}, \epsilon}$ that runs in $O((\log n)/ \epsilon)$ rounds in the \local model.
\end{theorem}
\addtocounter{theorem}{-1}
\end{NoHyper}
\endgroup
\begin{proof}

Consider Algorithm \ref{alg:local-1}. The number of rounds is clearly $O(r) = O(1/\eps \cdot \log n)$, since it is dominated by the steps in Lines \ref{ln:local-collect}, \ref{ln:local-collect-2}, and \ref{ln:local-collect-3} which is $O(r)$. Appealing to Lemma \ref{lem:low-diameter-ds}, if $\tilde{D} \leq D$, then there must be a subgraph $C$ with diameter at most $r$ whose density is at least $(1-\eps)D \geq (1-\eps)\tilde{D}$. Therefore, at least one vertex must be active. Among the active vertices, there must be a black vertex, i.e., the active vertex with the smallest ID. We therefore have a non-empty output.
 
The next observation is that if $u$ and $v$ are black vertices, then $H(v)$ and $H(u)$ are disjoint. Otherwise, there is a path of length at most $2r$ from $u$ to $v$ which leads to a contradiction. This is because $u$ and $v$ cannot both be black vertices if $\dist(u,v) \leq 2r$. Furthermore, it is easy to see that if two subgraphs $G[A]$ and $G[B]$, where $A,B \subset V$ are disjoint, have density at least $(1-\eps)\tilde{D}$, then $d(G[{A \cup B}]) \geq (1-\eps)\tilde{D}$. To see this, 
\[d(G[A \cup B]) \geq \frac{|E(A)| + |E(B)|}{|A|+|B|} \geq \frac{(1-\eps)\tilde{D}(|A|+|B|)}{|A|+|B|} = (1-\eps)\tilde{D}~.\]

Let the set of black vertices be $B$. Since we argued that $H(v)$'s are disjoint for  $v \in B$ and the output subgraph is $H=\cup_{v \in B} H(v)$,  we deduce that $d(H) \geq (1-\eps)\tilde{D}$. 

Finally, we need to argue that the algorithm is correct for when $\tilde{D}>D$. Note that if the output subgraph is non-empty, its density is at least $(1-\eps)\tilde{D}$. Hence, if $\tilde{D} > D$, the algorithm may output an empty subgraph or a subgraph with density $(1-\eps)\tilde{D}$ which are both acceptable. \end{proof}

We show that Theorem \ref{thm:main-local-1} is tight in terms of $\eps$ up to a constant. 

\begin{lemma} \label{lem:local-lowerbound}
Let $0 < \epsilon < 0.1$. Any (randomized)  algorithm that solves $\DenseSubgraph{\tilde{D}, \eps}$ correctly with probability at least 0.51 requires more than $1/(10\eps)$ rounds.
\end{lemma}
\begin{proof}
Consider deterministic algorithms for $\DenseSubgraph{\tilde{D},\eps}$ where $\tilde{D} = 1 - \eps$. Without loss of generality, assume $1/(10\eps)$ is an integer and let $\ell = 4/(10\eps)+1$.  Consider $\ell$ vertices $1,\ldots,\ell$ where $(i,i+1) \in E$ for $1 \leq i \leq \ell-1$. We consider two cases. Case 1: $(\ell,1) \in E$ in which the graph (called $G_1$) is a cycle. Case 2: $(\ell,1) \notin E$ in which the graph (called $G_2$) is a chain of $\ell$ vertices. Let $v = \floor{\ell/2}+1$.  In both cases, the $(1/(10\eps))$-neighborhood of $v$ is the same and therefore $h_v$ must be the same regardless of whether the network is $G_1$ or $G_2$. We observe that a chain of  $t$ vertices has density $1-1/t$.

If $h_v = 0$, then let  the underlying graph be $G_1$. Then $D=1 > 1-\eps = \tilde{D}$, and therefore the output must be non-empty. The only correct output is when every vertex outputs 1 since otherwise the output subgraph's density is at most $1-1/(4/(10\eps)+1) < (1-\eps)^2 = (1-\eps)\tilde{D}$ for $\epsilon < 0.1$. Hence, the algorithm fails.

If $h_v = 1$, then let  the underlying graph be $G_2$.  Then $D = 1-1/(4/(10\eps)+1)  <  (1-\eps)\tilde{D}$. Then, the only correct output is that every vertex outputs 0. Therefore, the algorithm fails.

For a randomized lower bound, we choose the above two inputs with probability 0.5 each and therefore the probability that a deterministic algorithm is correct is at most 0.5. By Yao's minimax principle, no randomized algorithm outputs correctly with probability more than 0.5.
\end{proof}

We remark that an interesting {\em open question} is whether the $\log n$ factor is necessary for either the lower bound or the upper bound.

\paragraph{The locality of the directed densest subgraph} In Appendix \ref{sec:directed-DS}, we show that similar locality result also holds for directed densest subgraph. 

%% file: congest.tex

\section{Dense Subgraph Detection in the CONGEST Model}\label{sec:congest}

\paragraph{Relating the densest subgraph problem and the lowest outdegree orientation problem} We show how to find  a) a dense  subgraph and b) a low outdegree orientation (in Section \ref{sec:loo}) by first solving the same feasibility LP. Let us first consider the  LP formulations in Figure \ref{fig:linear-programs}.

\begin{figure}[h]
\noindent
{\centering
\fbox{
\begin{minipage}[t][4.5cm]{.45\textwidth}
~
  (\Primal) Maximum Density Subgraph
  ~
\begin{align*}
\qquad \textup{maximize} \qquad & \sum_{e \in E} x_{e} \\
 \textup{subject to} \qquad & \\
& \sum_{v \in V} y_v = 1 \\
\forall e=uv \in E, \qquad &x_{e} \leq y_u \text{ and } x_{e} \leq y_v\\
\forall e \in E, u \in V, \qquad & y_{e}, x_{u} \geq 0 ~.
\end{align*}
\end{minipage}}
\fbox{
\begin{minipage}[t][4.5cm]{.45\textwidth}
 
 (\Dual) Lowest Outdegree Orientation
~
\begin{align*}
\qquad \textup{minimize} \qquad & z \\
 \textup{subject to} \qquad & \\
\forall e=uv \in E, \qquad & \alpha_{eu} + \alpha_{ev} \geq 1 \\
\forall e \in E, u \in V, \qquad & \sum_{e \ni u} \alpha_{eu} \leq z \\
\forall e=uv \in E, \qquad &  \alpha_{eu},\alpha_{ev}\geq 0~.
\end{align*}
\end{minipage}
}
}
\caption{Linear programs for densest subgraph and fractional lowest outdegree orientation.}\label{fig:linear-programs}
\end{figure}

Given a subgraph $H \subseteq G$ of size $k$, in the primal, we can set $y_v = 1/k$ for all $v \in V(H)$ and $x_{e}=1/k$ for all edges $e \in E(H)$ while setting other variables to 0. Then, $\sum_{e \in E} x_{e} = |E(H)|/k= d(H)$. In fact, the optimal value of the LP is exactly the maximum subgraph density $D$. Charikar gave a rounding algorithm that recovers the densest subgraph  \cite{Charikar00}. 

We observe that the dual models the lowest outdegree orientation problem. In particular, if an edge $e=uv$ is oriented from $u$ to $v$ then we set $\alpha_{eu}=1$ and $\alpha_{ev}=0$. By duality, the dual is fractionally feasible if and only if $z \geq D$.  

Now we consider the feasibility versions of the programs. We say a primal solution is a solution satisfying $\PRIMAL{z}$ if it is a feasible solution whose objective function is at least $z$. Similarly, we say a dual solution is a  solution satisfying $\DUAL{z}$ if it is a feasible dual solution whose objective function is at most $z$. 

\paragraph{Adapting the algorithm of Bahmani et al. \cite{BahmaniGM14}} Algorithm \ref{alg:fractional_dual} and Algorithm \ref{alg:mwu-densest-subgraph} are adapted from the multiplicative weights update approach of \cite{BahmaniGM14} for solving \Dual{} and \Primal. We modify them in a way so that we only have to operate with integers instead of real-valued weights. This is more suitable for the \congest model because it may take $\omega(1)$ rounds to transfer a real-valued weight over an edge. Similar to the tree packing method for minimum cuts \cite{Thorup07}, this is another example of a combinatorial algorithm  derived from multiplicative weights update method, where the weights are only used in the analysis but not in the algorithms. However, as mentioned in the first section, this adaptation has a dependence on $\diameterG$. We will subsequently show how to remove this dependence when solving $\DensestSubgraph{\tilde{D}, \epsilon}$.

Algorithm \ref{alg:fractional_dual} and Algorithm \ref{alg:mwu-densest-subgraph} have the same structure except for the outputs. Both algorithms consist of $T=O((\log n) /\epsilon^2)$ iterations. Every edge maintains a load $\ell(e)$ throughout the algorithms. In each iteration $t$, each vertex $u$ has $z$ units budget. Each vertex $u$ distributes the $z$ units of budget to the $\alpha^{(t)}_{eu}$ variables of the incident edges with the lowest $\lceil z/2 \rceil$ load. It allocates 2 units of budget to the incident edges with the lowest $\lceil z/2 \rceil - 1$ load and the remaining budget to the other edge. Then the load of an edge is updated by adding the allocation from both endpoints.  In Algorithm \ref{alg:fractional_dual}, we output the average of the allocations, $\sum_{t} \alpha^{(t)}_{eu} /T$, multiplied by $(1+2\epsilon)$.  In Algorithm \ref{alg:mwu-densest-subgraph}, in each iteration, the loads of the edges will be used as a guide to find a dense subgraph.

We can implement the algorithms by only sending integers across each link. This is because the value of each $\alpha_{eu}$ and $\ell(e)$ is a summation of an integer and an integer multiple of $z - 2(\lceil z/2 \rceil -1)$. A pair of integers will be enough to express each value involved in the algorithms.

 \begin{algorithm}
\caption{\textsc{Fractional\_Dual}($z$, $\epsilon$)}\label{alg:fractional_dual}
\begin{algorithmic}[1]\small
\State Let $T = K \cdot {(1/\eps^2) \cdot \log n }$ for some sufficiently large constant $K$.
\State Initialize the load $\ell(e) \leftarrow 0$ for each edge $e$.
\For{$t = 1, 2,\ldots , T$}
\For{each vertex $u$}
	\State Let $e_1, e_2, \ldots, e_{\deg(u)}$ be the edges adjacent to $u$ where $\ell(e_1) \leq \ell(e_2) \ldots \leq \ell(e_{\deg(u)})$.
	\State Set $\alpha^{(t)}_{e_i u} \leftarrow 2$ for $i = 1, \ldots, \min(\lceil z/2 \rceil - 1, \deg(u))$.
	\State Set $\alpha^{(t)}_{e_{\lceil z/2 \rceil u} } \leftarrow z - 2 \cdot (\lceil z/2 \rceil - 1)$ if $\deg(u) \geq \lceil z/2 \rceil$.
\EndFor
\For{each edge $e = uv$}
	\State Set $\ell(e) \leftarrow \ell(e) + \alpha^{(t)}_{eu} + \alpha^{(t)}_{ev}$. \label{ln:update-l-dual}
\EndFor
\EndFor
\State Return $\displaystyle \alpha_{eu} = \frac{\sum_{t=1}^T \alpha^{(t)}_{eu}}{T} \cdot (1+12\epsilon)$.
\end{algorithmic}
\end{algorithm}
 \begin{algorithm}
\caption{\textsc{Integral\_Primal}($z$, $\epsilon$)}\label{alg:mwu-densest-subgraph}
\begin{algorithmic}[1]\small
\State Let $T = K \cdot {(1/\eps^2) \cdot \log n }$ for some sufficiently large constant $K$.
\State Initialize the load $\ell(e) \leftarrow 0$ for all $e$.
\For{$t = 1, 2,\ldots , T$}
\For{each vertex $u$}
		\State Let $e_1, e_2, \ldots, e_{\deg(u)}$ be the edges adjacent to $u$ where $\ell(e_1) \leq \ell(e_2) \ldots \leq \ell(e_{\deg(u)})$.
	\State Set $\alpha^{(t)}_{e_i u} \leftarrow 2$ for $i = 1, \ldots, \min(\lceil z/2 \rceil - 1, \deg(u))$.
	\State Set $\alpha^{(t)}_{e_{\lceil z/2 \rceil u} } \leftarrow z - 2 \cdot (\lceil z/2 \rceil - 1)$ if $\deg(u) \geq \lceil z/2 \rceil$.
\EndFor
\State Let $\ell_{\min} = \min_{e \in E} \lfloor \ell(e) \rfloor$ and $\ell_{\max} = \ell_{\min} + \lceil \frac{1}{\epsilon}\log \frac{2m}{\epsilon} \rceil$. \label{ln:l-min} \For {each integer $\ell \in [\ell_{\min}, \ell_{\max}]$} \label{ln:checking_start}
\State Let $V'_\ell$ be the set of vertices with at  least $\lceil z/2 \rceil$ incident edges $e$ with $\lceil \ell(e) \rceil \leq \ell$.
\State Test if $G[V'_\ell]$ is a graph of density at least $(1-3\epsilon)z$. 
\State If yes, output $G[V'_\ell]$ and terminate.
\EndFor \label{ln:checking_end}

\For{each edge $e = uv$}
	\State Set $\ell(e) \leftarrow \ell(e) + \alpha^{(t)}_{eu} + \alpha^{(t)}_{ev}$. \label{ln:update-l-primal}

\EndFor
\EndFor
\end{algorithmic}
\end{algorithm}

We will show that if we run the algorithms on the same input with the same parameters $z$ and $\epsilon$, then at least one of the following must be true:
\begin{itemize} 
\item Algorithm \ref{alg:fractional_dual} returns a  solution satisfying $\DUAL{(1+12\epsilon)z}$. 
\item  Algorithm \ref{alg:mwu-densest-subgraph} outputs a subgraph of density at least $(1-3\epsilon)z$, i.e., an integral solution for $\PRIMAL{(1-3\epsilon)z}$.
\end{itemize}

Define $w_e = (1-\epsilon)^{\ell(e)}$ to be the weight associated with the edge $e$. We will use $\ell^{(t)}(e)$ and $w^{(t)}_e$ if we are referring to the load or the weight of an edge $e$ at the beginning of iteration $t$. We often use $u$ and $v$ to denote the endpoints of an edge $e$ when the context is clear. 

If it is the case that $\sum_{e \in E}w^{(t)}_e (\alpha^{(t)}_{eu} + \alpha^{(t)}_{ev}) \geq \sum_{e \in E} w^{(t)}_e$ for every iteration $t$, then we show that Algorithm \ref{alg:fractional_dual} returns a feasible solution for $\DUAL{(1+2\epsilon)z}$. The following lemma is a standard derivation of the multiplicative weights update method \cite{AHS12, Young95}. We defer the proof to the appendix.

\begin{lemma}\label{lem:dual_feasible}
Let $\epsilon > 0$ be a sufficiently small constant. Suppose that $\sum_{e \in E}w^{(t)}_e (\alpha^{(t)}_{eu} + \alpha^{(t)}_{ev}) \geq \sum_{e \in E} w^{(t)}_e$  for every iteration $t$ in Algorithm \ref{alg:fractional_dual}. Then the dual variables $\{\alpha_{eu}\}$ returned by Algorithm \ref{alg:fractional_dual} must be a feasible solution satisfying $\DUAL{(1+12\epsilon)z}$. 
\end{lemma}

Note that if we provide the algorithms with $z \geq D$, the condition $\sum_{e \in E}w^{(t)}_e (\alpha^{(t)}_{eu} + \alpha^{(t)}_{ev}) \geq \sum_{e \in E} w^{(t)}_e$ holds for every iteration $t$. This is because if $z \geq D$ then we know there exists a feasible solution such that $(\alpha_{eu} + \alpha_{ev}) \geq 1$ for every edge $e=uv$. This implies $\sum_{e \in E}w^{(t)}_e (\alpha_{eu} + \alpha_{ev}) \geq \sum_{e \in E} w^{(t)}_e$. Note that \[
\sum_{e \in E}w^{(t)}_e (\alpha_{eu}^{(t)} + \alpha_{ev}^{(t)} ) =  \sum_{u \in V}\sum_{ e \ni u } \alpha^{(t)}_{eu} w^{(t)}_e=\sum_{u \in V}\sum_{ e \ni u } \alpha^{(t)}_{eu}  (1-\epsilon)^{\ell^{(t)}(e)}~.\]
Hence, the way we assign $\{\alpha^{(t)}_{eu}\}$ maximizes $\sum_{e \in E}w^{(t)}_e (\alpha^{(t)}_{eu} + \alpha^{(t)}_{ev})$ over the set of feasible solutions $S(z)$, where 
\[S(z) = \left\{ \alpha :  \sum_{e \ni u} \alpha_{eu} \leq z \text{ for all $u$ and } 0 \leq \alpha_{eu} \leq 2 \text{ for all } \alpha_{eu} \right\} ~.
\]

Therefore,  $\sum_{e \in E}w^{(t)}_e (\alpha^{(t)}_{eu} + \alpha^{(t)}_{ev}) \geq \sum_{e \in E}w^{(t)}_e (\alpha_{eu} + \alpha_{ev}) \geq \sum_{e \in E} w^{(t)}_e$. From the above lemma, this implies Algorithm \ref{alg:fractional_dual} will output a fractional solution satisfying $\DUAL{z}$.

Next, we need to show that if the opposite holds, i.e., $\sum_{e \in E}w^{(t)}_e (\alpha^{(t)}_{eu} + \alpha^{(t)}_{ev}) < \sum_{e \in E} w^{(t)}_e$  for some  iteration $t$, then Algorithm \ref{alg:mwu-densest-subgraph} outputs a subgraph of density at least $(1-3\epsilon)z$.   First, we rely on  the following lemma which is a paraphrase of \cite{BahmaniGM14}. We include the proof in the appendix for the sake of completeness. 

\begin{lemma}[Paraphrase of \cite{BahmaniGM14}]\label{lem:subgraph}
Let $1/2 < F < 1$. Suppose there exists a set of (non-negative) weights $\{ w_e' \}$ such that  
\[\sum_{e \in E} w'_e > F \cdot \max_{\alpha \in S(z)} \big(  \sum_{e \in E}w'_e (\alpha_{eu} + \alpha_{ev})  \big)~.
\] 
Let $V''_\lambda$ denote the set of vertices that are incident to at least $\lceil z/2 \rceil$ edges $e$ with $w'_e \geq \lambda$. Then, there exists $\lambda$ such that $d(G[V''_\lambda]) > F \cdot z$. 
\end{lemma}

Hence, if there is an iteration $t$ where  $\sum_{e \in E}w_e^{(t)} (\alpha_{eu}^{(t)} + \alpha_{ev}^{(t)}) < \sum_{e \in E} w_e^{(t)}$, we can apply the above lemma with $F=1$ and $w_e' = w_e^{(t)}$ to find the desired dense subgraph. However, there may be $\Omega(m)$ potential values for $\lambda$ to check. We circumvent this by discretizing the weights. Since we are only testing the densities of $G[V'_\ell]$ for integer loads $\ell$, we are effectively discretizing the weights. We will show that this discretization only introduces a small error on the inequality so that we can apply Lemma \ref{lem:subgraph} with $F =1-O(\epsilon)$.  This, in turn, gives us the following.

\begin{lemma}\label{lem:primal_feasible}
Suppose that there exists an iteration $t$ in Algorithm \ref{alg:mwu-densest-subgraph} where $\sum_{e \in E}w_e^{(t)} (\alpha_{eu}^{(t)} + \alpha_{ev}^{(t)}) < \sum_{e \in E} w_e^{(t)}$. Then for $0< \epsilon < 1/4$, Algorithm \ref{alg:mwu-densest-subgraph} will output a subgraph of density at least $(1-3\epsilon)z$ in that iteration.
\end{lemma}

\begin{lemma}\label{theorem:rounding-DS}
Suppose $0<\epsilon < 1/4$ and $\tilde{D} \leq D$. Then, there exists a deterministic algorithm that solves $\DenseSubgraph{\tilde{D},\eps}$ in $O(\diameterG +   1/\eps^3 \cdot \log^2 n)$ rounds in the \congest model.
\end{lemma}

\begin{proof}
Assume that $\tilde{D} \leq D$. Let $z = (1-\epsilon/2) \tilde{D}$ and $\epsilon' = \epsilon/8$. If we were to run Algorithm \ref{alg:fractional_dual} with parameters $z$ and $\epsilon'$, then it cannot possibly output a feasible solution for $\DUAL{z (1+2\epsilon')}$, since $\DUAL{z(1+2\epsilon')}$ is infeasible. This is because $z(1+2\epsilon') = (1-\epsilon/2)(1+\epsilon/4) \tilde{D} < \tilde{D} \leq D$. This implies that Algorithm \ref{alg:mwu-densest-subgraph} has to output a subgraph with density of at least $(1- 3\epsilon')z \geq (1-\epsilon/2) (1-\epsilon /2) \tilde{D} \geq (1-\epsilon) \tilde{D}$. 

We now argue that Algorithm \ref{alg:mwu-densest-subgraph} can be implemented in $O(\diameterG +   1/\eps^3 \cdot \log^2 n)$ rounds in the \congest model. Each step in the main loop of Algorithm \ref{alg:mwu-densest-subgraph} uses $O(1)$ rounds except for Line \ref{ln:l-min} and the inner loop (Line \ref{ln:checking_start} -- Line \ref{ln:checking_end}). Line \ref{ln:l-min} can be done in $O(\diameterG)$ rounds. The inner loop tests the density of $O((1/\epsilon)\log n)$ subgraphs. In total, there are $O((1/\epsilon)\log n) \cdot O((1/\epsilon^2) \log n)$ subgraphs to be tested. Each subgraph and its density can be computed in $O(\diameterG)$ rounds. These steps can be pipelined to run in $O(\diameterG + (\log^2 n) /\epsilon^3)$ rounds. Finally, we argue  that each $\alpha_{eu}^{(t)}$ can be represented using $O(\log \log n + \log (1/\epsilon))$ bits (Lemma \ref{lem:bits}) and therefore in each iteration, Line \ref{ln:update-l-primal} in Algorithm \ref{alg:mwu-densest-subgraph} (and Line \ref{ln:update-l-dual} in Algorithm \ref{alg:fractional_dual}) can be done in one \congest round.
\end{proof}

Using the above lemma, we now present the following results:
\begin{itemize}
\item There exists an algorithm that solves  $\DenseSubgraph{\tilde{D},\eps}$ in $O(\poly(\log n, 1/\eps))$ rounds w.h.p. In particular, we are able to avoid $\diameterG$ rounds in Lemma \ref{theorem:rounding-DS}. 
\item There exists an algorithm that finds a $1-\eps$ approximation of the densest subgraph (instead of solving the parameterized version) in $O(\diameterG + \poly(\log n, 1/\eps))$ rounds w.h.p. 
\end{itemize}

\paragraph{Solving $\DenseSubgraph{\tilde{D},\eps}$} Our first goal is to avoid $O(\diameterG)$ rounds to solve the parameterized version of the densest subgraph problem. The main idea is to apply the low diameter decomposition and solve the problem in each component. First, we recall that the low-diameter decomposition can be implemented efficiently in the \congest model. See also \cite{CS2019}.

\begin{lemma}  [\cite{MillerPX13}]\label{thm:low-diameter-decomposition}
There exists an algorithm that decomposes the graph into disjoint components such that: 1) Each edge is inter-component with probability at most $\eps$, 2) Each component has diameter $O(1/\eps \cdot \log n)$ w.h.p. and 3) Runs in $O( 1/\eps \cdot \log n)$ rounds in the \congest model w.h.p.
\end{lemma}

In the decomposition given by Miller et al. \cite{MillerPX13}, the rough idea is that each vertex $v$ draws $\delta_v$ from the exponential distribution $Exp(\eps)$. Let $\delta = K/\eps \cdot \log n$ for some sufficiently large constant $K$. At time step $\floor{\delta - \delta_v}$, $v$ wakes up and starts a breadth first search (BFS) if it has not been covered by another vertex's BFS. At the end of the decomposition, $v \in \textup{cluster}(u)$ if $u = \min_{y \in V}(\dist(y,v) + \floor{\delta - \delta_y})$. This algorithm can be simulated in the \congest model in $O( 1/\eps \cdot \log n)$ rounds. Now, we are ready to prove our next main result.

\begingroup
\begin{NoHyper}
\def\thetheorem{\ref{thm:main-congest-1}}
\begin{theorem}
There exists a randomized algorithm that solves $\DensestSubgraph{\tilde{D}, \epsilon}$ w.h.p.~and runs in $O((\log^3 n) / \epsilon^3)$ rounds in the \congest model.
\end{theorem}
\addtocounter{theorem}{-1}
\end{NoHyper}
\endgroup

\begin{proof} 

We apply the decomposition above  with parameter $\eps/2$ to the graph to obtain low-diameter components $C_1,\ldots,C_t$.  Let $H^*$ be the densest subgraph with $n^*$ vertices and $m^*$ edges. We condition on the event $\xi$ that at most $ \eps |E(H^*)|$ edges in $E(H^*)$ are inter-component. This happens with probability at least $1/2$ according to  Markov's inequality. 

First, consider the case $\tilde{D} \leq D$. Using a similar argument as in the proof of Lemma \ref{lem:low-diameter-ds}, we can show that the densest subgraph of at least one component has density at least $(1-\eps)D \geq (1-\eps)\tilde{D}$. Specifically, let $H_i^*  = H^* \cap C_i$ and let $C_i^* \subseteq C_i$ be the densest subgraph in $C_i$. Furthermore, let $|V(H_i^*)|=n^*_i$ and $|E(H_i^*)|=m^*_i$.

If $d(C_i^*) < (1-\eps) D$ for all $1 \leq i \leq t$. Then, $d(H_i^*) \leq d(C_i^*) < (1-\eps)D$. This implies $m_i^* < (1-\eps) m^* n_i^*/n^*$.  Therefore, $\sum_{i=1}^t m_i^* < (1-\eps) m^*$ which means that more than $\eps m^*$ edges in $E(H^*)$ are inter-component. This is  a contradiction since we condition on $\xi$. 

For each low-diameter component $C_i$, using the algorithm in Lemma \ref{theorem:rounding-DS}, we can solve $\DenseSubgraph{\tilde{D},\eps}$ in $C_i$. As argued above, since $\tilde{D} \leq D$, we must have a non-empty output in some component $i$. Observe that since the diameter of each $C_i$ is $O(1/\eps \cdot \log n)$, this takes $O(1/\eps^3  \cdot \log^2 n)$ rounds.

We ensure $\xi$ happen w.h.p. by repeating $O(\log n)$ trials. The total number of rounds becomes $O(1/\eps^3  \cdot \log^3 n)$.
However, there is a catch regarding the consistency of the output. Recall that we want the  subgraph induced by the marked vertices $\{ v : h_v = 1 \}$ has density at least $(1-\eps) \tilde{D}$. It is possible that for different trials $j$ and $j'$, $v$ might be marked and unmarked respectively. We need to address the issue of how to decide the final output.

This can be done with proper bookkeeping. Originally, all vertices are unmarked. In each trial, compute a low-diameter decomposition $C_1, C_2, \ldots, C_t$. For each low-diameter component $C_i$, if it does not contain a marked vertex, check if there exists a subgraph  $H_i \subseteq C_i$ with density at least $ (1-\eps)\tilde{D}$ using the algorithm in Lemma \ref{theorem:rounding-DS}.  If there exists such subgraph, mark every vertex in $H_i$. In the end, return the subgraph induced by the marked vertices. The output must be non-empty w.h.p. by considering the first time the event $\xi$ occurs. Since all the output subgraphs among the trials are disjoint and they have density at least $(1-\eps)\tilde{D}$, their union must have density at least $(1-\eps)\tilde{D}$ as argued in the proof of Theorem \ref{thm:main-local-1}. 
 
In the case $\tilde{D} > D$,  the algorithm either returns a subgraph with  density at least $(1-\eps)\tilde{D}$ or an empty subgraph which are both acceptable.
\end{proof}

\paragraph{Approximating the densest subgraph} We can also apply the ideas above to find a $1-\eps$ approximation of the densest subgraph. This is done by simply running the algorithm in Theorem \ref{thm:main-congest-1} on different guesses for the maximum subgraph density $D$.

\begingroup
\begin{NoHyper}
\def\thetheorem{\ref{cor:main-congest-2}}
\begin{corollary}
There exists a randomized algorithm that finds a $(1-\epsilon)$-approximation to the maximum density subgraph w.h.p.~and runs in $O(\diameter{G} + (\log^4 n/ \epsilon^4) )$ in the \congest model.
\end{corollary}
\addtocounter{theorem}{-1}
\end{NoHyper}
\endgroup

\begin{proof} 
For each $\tilde{D} =1, (1+\eps),(1+\eps)^2,\ldots, (1+\eps)^{\lceil \log_{1+\eps} n \rceil}$, we run the algorithm in Theorem \ref{thm:main-congest-1}. This requires $O(1/\eps^4 \cdot \log^4 n )$ rounds.  We then identify the largest $\tilde{D}$ in which we have a non-empty output.  In particular, we refer to when $\tilde{D} = (1+\eps)^i$ as phase $i$. Each vertex $v$ sets $\psi(v) = i$ where $i$ is the largest phase in which it is marked.  In the end, we can broadcast $j = \max_{v \in V} \psi(v)$ to all vertices $v$ in $O(\diameterG)$ rounds.  Then, for every vertex $v$, if $v$ is marked in phase $j$, set $h_v = 1$. \end{proof}

We include the following lemma, which will be useful in bounding the running time for rounding the fractional dual solution into an integer dual solution in the next section.

\begin{lemma}\label{lem:bits}
Suppose that we run Algorithm \ref{alg:fractional_dual} with parameters $z$ and $\epsilon$, where $z$ is an integer and $\epsilon$ is a (negative) power of 2. Moreover, assume, $T = \Theta((\log n) / \epsilon^2)$, in the main loop of Algorithm \ref{alg:fractional_dual}, is a power of 2. Each $\alpha_{eu}$ in the solution returned by Algorithm \ref{alg:fractional_dual} contains at most $O(\log \log n  + \log(1/\epsilon))$ bits.
\end{lemma}
%
%

%% file: CONGEST_splitting.tex
\sloppy
\newcommand{\Mod}[1]{\ \mathrm{mod}\ #1}

\section{Deterministic Algorithms for Low Outdegree Orientation}\label{sec:loo}

Recall that in the low outdegree orientation problem, we are given an integer parameter $\tilde{D} \geq D$. The goal is to find an orientation of the edges such that for every vertex, the number of outgoing edges is upper bounded by $(1+\epsilon)\tilde{D}$.

In Section \ref{sec:congest}, we showed that Algorithm \ref{alg:fractional_dual} can be used to obtain a  solution for the fractional version of the low outdegree orientation problem.  In this section, we will show how to round a fractional solution to an integral solution deterministically in the \textsf{CONGEST} model. We first present the framework and then describe the subroutine for the directed splitting procedure adapted from \cite{GHKMSU17} in Section \ref{sec:splitting}.

Let $\epsilon_1, \epsilon_2 \in \Theta(\epsilon)$ be  error control parameters which will be determined later. First, we will run Algorithm \ref{alg:fractional_dual} with parameters $\tilde{D}$ and $\epsilon_1/12$ to obtain a fractional solution for $\DUAL{(1+\epsilon_1)\tilde{D}}$. Let $\{ \alpha'_{eu} \}_{e \in E, u \in V}$ be the output fractional solution. They satisfy the conditions that $\alpha'_{eu} + \alpha'_{ev} \geq 1$ for every edge $e = uv$ and $\sum_{e \ni u}\alpha'_{eu} \leq (1+\epsilon_1) \tilde{D}$ for every vertex $u$. We show how to round the $\alpha'$-values to $\{0,1\}$ bit-by-bit deterministically and incur bounded errors on the constraints.

We may also assume without loss of generality that both $1/\epsilon_1$ and the number of iterations, $T = \Theta((\log n) / \epsilon_{1}^2)$, in the main loop of Algorithm \ref{alg:fractional_dual}, is a power of 2. We can apply Lemma \ref{lem:bits} to obtain an upper bound, $B = O(\log \log n + \log(1/\epsilon_1))$, on the number of bits needed to store each $\alpha'_{eu}$. In the case where the maximum degree $\Delta$ is very small, we can even truncate the bits without creating much error. Let $t = \min(B, \lceil \log_{2} (\Delta / \epsilon_2) \rceil)$. We round the $\alpha'$-values up to the $t$'th bit after the decimal point.  In other words, we set $\alpha^{(0)}_{eu} = \lceil 2^{t} \cdot \alpha'_{eu} \rceil / 2^{t}$ for every variable $\alpha_{eu}$. If $t = B$, we are just setting $\alpha^{(0)}_{eu} = \alpha'_{eu}$. 

Note that in the case $t =\lceil \log_{2} (\Delta / \epsilon_2) \rceil$, because $\tilde{D} \geq 1$, we have 
\[\sum_{e \ni u}\alpha^{(0)}_{eu} \leq \sum_{e \ni u}\alpha_{eu}' + \deg(u) \cdot 2^{- \lceil \log_{2} (\Delta  /\epsilon_2) \rceil } \leq (1+\epsilon)\tilde{D} + \epsilon_2\leq (1+\epsilon_1)(1+\epsilon_2)\tilde{D}~.  \]

The algorithm (Algorithm \ref{alg:dualrounding}) consists of $t$ iterations. It processes the $\alpha$-values bit-by-bit, from the $t$'th bit to the first bit after the decimal point. When it is processing the $k$'th bit, for each $\alpha_{eu}$, we will round its $k$'th bit either up or down. Therefore, after we have processed the first bit in the last iteration, all the $\alpha$-values are integers. Let $\alpha_{eu}(i)$ be the $(t-i+1)$'th bit of $\alpha_{eu}$ after the decimal point (i.e.~the $i$'th rightmost bit after the initial rounding).

Let $\alpha^{(k)}_{eu}$ denote the value $\alpha_{eu}$ at the end of iteration $k$. During iteration $k$, if $\alpha^{(k-1)}_{eu}(k) = 1$, we will either need to round it up (set $\alpha^{(k)}_{eu} = \alpha^{(k-1)}_{eu} + 2^{-(t-k+1)} $) or round it down (set $\alpha^{(k)}_{eu} = \alpha^{(k-1)}_{eu} - 2^{-(t-k+1)} $) so that $\alpha^{(k)}_{eu}(k) = 0$.

Consider an edge $e = uv$.  If $\alpha^{(k-1)}_{eu}(k) = 0$ or $\alpha^{(k-1)}_{ev}(k) = 0$, we will round both of their $k$'th bit down so that $\alpha^{(k)}_{eu}(k) = 0$ and $\alpha^{(k)}_{ev}(k) = 0$. All the remaining edges must be contained in the graph $G_{k} = (V, E_{k})$, where $E_k = \{uv \mid \alpha^{(k-1)}_{eu}(k) =  1 \mbox{ and } \alpha^{(k-1)}_{ev}(k) =  1\} $.  We will perform a deterministic directed splitting algorithm on $G_k$ which we adapt from \cite{GHKMSU17} to the \textsf{CONGEST} model. Let $\deg_{k}(u)$ denote the degree of $v$ in $G_{k}$. The outcome of the algorithm is an orientation of the edges in $E_{k}$ such that for each vertex $u$, $|\outdeg_{k}(u) - \indeg_{k}(u)| \leq \epsilon_3 \deg_{k}(u) + 12$.

Suppose that $e = uv$ is oriented from $u$ to $v$. We will round $\alpha^{(k-1)}_{eu}$ up and round $\alpha^{(k-1)}_{ev}$ down. We do the opposite if it is oriented from $v$ to $u$.

 \begin{algorithm}
\caption{Deterministic Rounding for Low Outdegree Orientation}\label{alg:dualrounding}
\begin{algorithmic}[1]\small
\State Obtain a feasible $\{\alpha'_{eu} \}_{eu}$ for $\DUAL{(1+\epsilon_1)\tilde{D}}$.
\State Set $\alpha^{(0)}_{eu} = \lceil 2^{t} \cdot \alpha'_{eu} \rceil / 2^{t}$ for every $\alpha_{eu}$.
\For{$k = 1 \ldots t$}
	 \For{every edge $e = uv$ s.t.~$\alpha^{(k-1)}_{eu}(k) = 0$ or $\alpha^{(k-1)}_{ev}(k) = 0$}
	 	 	\State Set $\alpha^{(k)}_{eu} =  \alpha^{(k-1)}_{eu}$ and then set $\alpha^{(k)}_{eu}(k) = 0$.  \Comment{round the $k$'th bit down}
	 	 	\State Set $\alpha^{(k)}_{ev} =  \alpha^{(k-1)}_{ev}$ and then set $\alpha^{(k)}_{ev}(k) = 0$.

	 \EndFor
	 
	 \State Let $G_k = (V,E_k)$, where $E_k = \{uv \mid \alpha^{(k-1)}_{eu}(k) =  1 \mbox{ and } \alpha^{(k-1)}_{ev}(k) =  1\} $.
	 \State Obtain a directed splitting of $G_k$ whose discrepancy is at most $\epsilon_3 \deg_k(u) + 12$ for each vertex $u$.

	 \For{every edge $e = uv$ where $u$ is oriented toward $v$}
	 	\State Set $\alpha^{(k)}_{eu} = \alpha^{(k-1)}_{eu} +  2^{-(t -k  +1)} $. \Comment{round up}
	 	\State Set $\alpha^{(k)}_{ev} =  \alpha^{(k-1)}_{ev} - 2^{-(t -k  +1)}$. \Comment{round down}
	 \EndFor
\EndFor

\For{each $e=uv$}
	\State Set $\alpha_{eu} = \min(\alpha^{(t)}_{eu}, 1)$.
\EndFor
\end{algorithmic}
\end{algorithm}

\begin{lemma}
Suppose that $\epsilon_3 \leq 1/4$. The $\{\alpha_{ev}\}$ values produced by Algorithm \ref{alg:dualrounding} satisfy the following properties. (1)  $\alpha_{eu} \in \{0, 1\} $ (2) For every edge $e = uv$, $\alpha_{eu} + \alpha_{ev} \geq 1$. (3) For every vertex $u$, $\sum_{e \ni u} \alpha_{eu}  \leq (1+\epsilon_1)(1+\epsilon_2)(1+\epsilon_3)^{t}\tilde{D} + 16$.
\end{lemma}

\begin{proof}
For (1), since we either round $\alpha^{(k-1)}_{ev}(k)$ up or down during iteration $k$, we have $\alpha^{(k)}_{ev}(k) = 0$ at the end of iteration $k$. Moreover, once $\alpha^{(k)}_{ev}(k)$ becomes $0$, $\alpha^{(k')}_{ev}(k)$ remains 0 for $k' \geq k$. Therefore, at the end of iteration $t$, we have $\alpha^{(t)}_{ev}(i) = 0$ for $1 \leq i \leq t$. This implies $\alpha_{ev}^{(t)}$ is a non-negative integer. Since the final output, $\alpha_{ev}$, is the minimum of 1 and $\alpha_{ev}^{(t)}$, we have $\alpha_{ev} \in \{0, 1\}$.

We show (2) inductively. In the beginning, since $\alpha^{(0)}_{eu} \geq \alpha'_{eu}$, we must have $\alpha^{(0)}_{eu} + \alpha^{(0)}_{ev} \geq 1$. Suppose that $\alpha^{(k-1)}_{eu} + \alpha^{(k-1)}_{ev} \geq 1$. During iteration $k$, if $\alpha^{(k-1)}_{eu}(k) = 0$ and $\alpha^{(k-1)}_{ev}(k) = 0$, then $\alpha^{(k)}_{eu} = \alpha^{(k-1)}_{eu}$ and $\alpha^{(k)}_{ev}  = \alpha^{(k-1)}_{ev}$ and so $\alpha^{(k)}_{eu} + \alpha^{(k)}_{ev} \geq 1$.  

If $\alpha^{(k-1)}_{eu}(k) = 1$ and $\alpha^{(k-1)}_{ev}(k) = 0$, then it must be the case that $\alpha^{(k-1)}_{eu}+ \alpha^{(k-1)}_{ev} \geq 1 + 2^{-(t-k+1)}$. After rounding $\alpha^{(k-1)}_{eu}$ down, we still have $\alpha^{(k)}_{eu} + \alpha^{(k)}_{ev}  \geq 1$. The case for $\alpha^{(k-1)}_{eu}(k) = 0$ and $\alpha^{(k-1)}_{ev}(k) = 1$ is symmetric. 

The remaining case is when $\alpha^{(k-1)}_{eu}(k) = 1$ and $\alpha^{(k-1)}_{ev}(k) = 1$. In this case, $e \in G_k$. We must have $\alpha^{(k)}_{eu} + \alpha^{(k)}_{ev} = \alpha^{(k-1)}_{eu} + \alpha^{(k-1)}_{ev}   \geq 1$, since one of them is rounded up and the other is rounded down.

For (3), let $D_0 = (1+\epsilon_1)(1+\epsilon_2)\tilde{D}$ and $D_{k} = (1+\epsilon_3) D_{k-1} + 12\cdot 2^{-(t-k+1)}$ for $k\geq 1$. We will show by induction that $\sum_{e \ni u} \alpha^{(k)}_{eu} \leq D_k$. For the base case, initially, we argued that $\sum_{e \ni u} \alpha^{(0)}_{eu}  \leq (1+\epsilon_1)(1+\epsilon_2)\tilde{D}$. For $k \geq 1$, note that the increase on the quantity $\sum_{e \ni u} \alpha_{eu}$ during iteration $k$ is at most  $2^{-(t-k+1)} \cdot (\epsilon_3 \cdot \deg_{k}(u) + 12)$. Therefore,
\begin{align*}
&\sum_{e \ni u} \alpha^{(k)}_{eu} \leq   2^{-(t-k+1)} \cdot (\epsilon_3 \cdot \deg_{k}(u) + 12) + \sum_{e \ni u} \alpha^{(k-1)}_{eu} \\
& \leq  \epsilon_3 D_{k-1} + 12 \cdot 2^{-(t-k+1)}  + D_{k-1} \leq (1+\epsilon_3) D_{k-1} + 12 \cdot  2^{-(t-k+1)} = D_{k}~.
\end{align*}
The second inequality follows since $2^{-(t-k+1)}\deg_{k}(u) \leq \sum_{e \ni u} \alpha^{(k-1)}_{eu} \leq D_{k-1}$. This completes the induction. Since $\epsilon_3 \leq 1/4$,  in the end, we have
\begin{align*}
\sum_{e \ni u} \alpha_{eu} \leq \sum_{e \ni u} \alpha^{(t)}_{eu} \leq D_t &\leq (1+\epsilon_1)(1+\epsilon_2)(1+\epsilon_3)^{t}\tilde{D} + 12\cdot (1/2) \cdot \sum_{k=0}^{t-1} \left(\frac{1+ \epsilon_3}{2}\right)^{k}\\
&\leq (1+\epsilon_1)(1+\epsilon_2)(1+\epsilon_3)^{t}\tilde{D} + 16 ~. \qedhere
\end{align*}
\end{proof}

\begingroup
\begin{NoHyper}
\def\thetheorem{\ref{thm:lowdeg}}
\begin{theorem}
Given an integer $\tilde{D} \geq D$, for any $32/\tilde{D} \leq \epsilon \leq 1/4$, there exists a deterministic algorithm in the \textsf{CONGEST} model that computes a $(1+\epsilon)\tilde{D}$-orientation and runs in
$$O\left(\frac{\log n}{\epsilon^2} + \left(\min(\log \log n, \log \Delta)+  \log(1/ \epsilon)\right)^{2.71}  \cdot (1/\epsilon)^{1.71}\cdot \log^2 n \right)  \leq  \tilde{O}((\log^2 n) /\epsilon^2)   \mbox{ rounds}.$$
\end{theorem}
\addtocounter{theorem}{-1}
\end{NoHyper}
\endgroup

\begin{proof}
We set parameters $\epsilon_1 = \epsilon_2 = \epsilon/8$ and $\epsilon_3 = \epsilon / (4t)$. Run Algorithm \ref{alg:dualrounding} to obtain integral $\{\alpha_{eu} \}$ that satisfy $\alpha_{eu} + \alpha_{ev} \geq 1$ for every $e=uv$ and $\sum_{e \ni u} \alpha_{eu} \leq(1+\epsilon_1)(1+\epsilon_2)(1+\epsilon_3)^{t}\tilde{D} + 16$ for every $u$. For each edge $e=uv$, if $\alpha_{eu} = 1$ then we orient $e$ from $u$ to $v$. Otherwise, we orient $e$ from $v$ to $u$. Since $16 \leq \epsilon \tilde{D}/2$, the out-degree of each vertex is upper bounded by:
\begin{align*}
\sum_{e \ni u} \alpha_{eu} &\leq(1+\epsilon_1)(1+\epsilon_2)(1+\epsilon_3)^{t}\tilde{D} + 16 \leq (1+\epsilon/2)\tilde{D} + 16 \leq (1+\epsilon)\tilde{D}~. 
\end{align*}

The running time for Algorithm \ref{alg:dualrounding} consists of the following. The number of rounds to obtain a fractional solution is $O((\log n) / \epsilon_{1}^2) = O((\log n) / \epsilon^2)$. Then, it consists of $t = \min(\log(\Delta / \epsilon_2),B) = O( \min(\log \log n, \log \Delta) + \log(1/\epsilon))$ iterations. Each iteration invokes a directed splitting procedure that runs in $O((1/\epsilon_3)^{1.71} \cdot \log^2 n )$ rounds by Theorem \ref{thm:splitting}. Recall that $\epsilon_3 = \epsilon/(4t)$,  the total number of rounds is therefore:

\begin{align*}
&O\left(\frac{\log n}{\epsilon^2} + t \cdot (1/\epsilon_3)^{1.71}\cdot \log^2 n \right) =O\left(\frac{\log n}{\epsilon^2} + t^{2.71} \cdot (1/\epsilon)^{1.71}\cdot \log^2 n \right) \\
&=O\left(\frac{\log n}{\epsilon^2} + \left(\min(\log \log n, \log \Delta) +  \log(1/ \epsilon)\right)^{2.71} \cdot (1/\epsilon)^{1.71}\cdot \log^2 n \right)~.  \qedhere
\end{align*}
\end{proof}

\subsection{Distributed Splitting in the CONGEST Model}\label{sec:splitting}

Given a graph $G=(V,E)$, a {\it weak $f(v)$-orientation} of $G$ is an orientation of the edges in $G$ such that there are at least $f(v)$ outgoing edges for each $v \in V$. In order to adapt the algorithm of \cite{GHKMSU17} for directed splitting, we need an algorithm for weak $\lfloor \deg(v) / 3 \rfloor$-orientation in the \textsf{CONGEST} model. The previous algorithms \cite{GHKMSU17,GS17} for weak $\lfloor \deg(v) / 3 \rfloor$-orientation  requires finding short cycles for containing each edge. They are not  adaptable to the \textsf{CONGEST} model. We use an augmenting path approach for finding a weak $\lfloor \deg(v) / 3 \rfloor$-orientation instead.

\begin{lemma}\label{lem:sinkless} There exists a deterministic distributed algorithm that computes a weak $\lfloor \deg(v) / 3 \rfloor$-orientation in $O(\log^2 n)$ rounds in the \congest model.
\end{lemma}

\begin{proof}
First we construct a new graph $G'$ as follows: Split every vertex $v$ into $\lceil \deg(v) /3 \rceil $ copies. Attach evenly the edges to each copy of the vertex so every copy except possibly the last gets 3 edges and the last copy gets $\deg(v) - 3 \cdot (\lceil \deg(v) /3 \rceil - 1)$ edges. Given an orientation of $G'$, we call a vertex $v$ a {\it sink} {\bf if it has exactly 3 incoming edges}. Clearly, a sinkless orientation (i.e.~an orientation where there are no sinks) in $G'$ corresponds to a weak $\lfloor \deg(v) / 3 \rfloor$-orientation of $G$. Moreover, one round in $G'$ can be emulated in $G$ by using one round.

Now we start with an arbitrary orientation of $G'$.  Some vertices might be sinks. We will use an augmenting path approach to eliminate sinks.

Divide the vertices into the following three types. Type I vertices are the sinks. Type II vertices are those $u$ such that $\deg(u) = 3$ and $\indeg(u) = 2$. Type III vertices are those with $\indeg(u) \leq 1$ or  $\deg(u) < 3$.

An augmenting path is a path $P=(u_1, \ldots, u_l)$ such that:
\begin{enumerate}
\item  $u_1$ is a Type I vertex. $u_l$ is a Type III vertex. $u_i$ is a Type II vertex for $1 < i < l$.

\item $u_{i+1}$ is oriented towards $u_i$ for $1 \leq i < l$.
\end{enumerate}
If $P$ is an augmenting path then flipping $P$ will make $u_1$ no longer a sink. Moreover, it will not create any new sink.

Consider a Type I vertex $u$. An augmenting path of length $O(\log n)$ starting from $u$ can be found as follows. Let $L_0 = \{u\}$. Given $L_{i-1}$, let $L_{i}$ be the set of incoming neighbors of every vertex in $L_{i-1}$. If $L_i$ contains a Type III vertex, then an augmenting path is found. Otherwise, it must be the case that $L_i$ contains all Type II vertices. Since there are at least $2\cdot|L_{i-1}|$ incoming edges from $L_{i}$ to $L_{i-1}$ ($\indeg(u)=2$ for $u  \in L_{i-1}$)  and Type II vertices have out-degree 1, we have $|L_i| \geq 2 \cdot |L_{i-1}|$. Hence, this process can only continue for at most $O(\log n)$ times. 

Every Type I vertex would be able to find an augmenting paths of length $O(\log n)$ this way. Moreover, note that these augmenting paths can only overlap at their ending vertex, since the intermediate Type II vertices have out-degree 1. Each ending vertex is a Type III vertex, which can only be the ending vertex of at most 3 augmenting paths since it has at most 3 outgoing edges. It selects an arbitrary augmenting path to accept. Therefore, at least 1/3 fraction of augmenting paths will be accepted. We flip along the accepted augmenting paths to fix Type I vertices so that we eliminate at least 1/3 fraction of the sinks. Therefore, it takes $O(\log n)$ repetitions to eliminate all the sinks. The total number of rounds is $O(\log n \cdot \log n) = O(\log^2 n)$.  The process can be easily implemented in the \textsf{CONGEST} model.
\end{proof}

Before we describe how to adapt the splitting algorithm of \cite{GHKMSU17}, we need to introduce the following definition. Given a function $\delta: V \to \mathbb{R}_{\geq 0}$ and $\lambda \in  \mathbb{Z}_{\geq 0}$. A {\it $(\delta, \lambda)$-path decomposition} $\mathcal{P}$ is a partition of $E$ into paths $P_1, \ldots, P_{\rho}$ such that \begin{enumerate} \item Each vertex $v$ is an endpoint of at most $\delta(v)$ paths. \item Each path $P_i$ is of length at most $\lambda$. \end{enumerate}

Given a  $(\delta, \lambda)$-path decomposition $\mathcal{P}$. The virtual graph $G_{\mathcal{P}} = (V,E_{\mathcal{P}})$ consists of exactly $\rho$ edges, where each path $P_i = (v_{i,start}, \ldots, v_{i, end})$ corresponds to an edge $(v_{i,start}, v_{i,end})$ in $E_{\mathcal{P}}$. Lemma \ref{lem:boost} is the adaption of \cite[Lemma 2.11]{GHKMSU17} to the \textsf{CONGEST} model. 

\begin{lemma}\label{lem:boost} Assume that $T(n,\Delta) \geq \log n$ is the running time of an algorithm $\mathcal{A}$ that finds a weak $\lfloor \deg(v) / 3\rfloor$-orientation in the \congest model. Then for any positive integer $i$, there is a deterministic distributed algorithm $\mathcal{A}$ that finds a $(\left(\frac{2}{3} \right)^{i} \cdot \deg(v) + 12, 2^i)$-path decomposition $\mathcal{P}$ in time $O(2^i \cdot T(n,\Delta))$ in the \congest model. \end{lemma}

The reason of why such an adaptation works is due to the fact that all the paths in $\mathcal{P}$ are disjoint. Therefore, one round  in $G_{\mathcal{P}}$ can be simulated using $O(\lambda)$ rounds in $G$ in the \textsf{CONGEST} model, given $\mathcal{P}$ is a $(\delta, \lambda)$-path decomposition.

For completeness, we explain how the path-decomposition in Lemma \ref{lem:boost} can be obtained. Let $\mathcal{P}_0$ denote the initial path decomposition where each path is a single edge in $E$. Given $\mathcal{P}_{i-1}$, $\mathcal{P}_i$ can be built as follows: Obtain a $\lfloor \deg(v) / 3 \rfloor$-orientation on $G_{\mathcal{P}_{i-1}}$. For each vertex $u$, group the outgoing edges into at least $\lfloor \lfloor \deg(u) / 3 \rfloor / 2 \rfloor$ pairs. For each such edge pair $(u, x)$ and $(u,w)$, we reverse the path that corresponds to $(u,x)$ and append it with $(u,w)$. The new degree of $u$ becomes at most $\deg(u) - 2\cdot \lfloor \lfloor \deg(u) / 3 \rfloor / 2 \rfloor \leq \frac{2}{3}\deg(u) +4	$. Therefore, if $\mathcal{P}_{i-1}$ is a $(\delta, \lambda(v))$-decomposition then $\mathcal{P}_{i}$ is a $(\frac{2}{3}\cdot \delta(v) + 4, 2\lambda(v))$ decomposition. 

Since $\mathcal{P}_0$ is a $(\deg(v), 1)$-path decomposition, $\mathcal{P}_i$ a $(z_i(v),2^{i})$-path decomposition, where 

$$z_i(v) = \left(\frac{2}{3}\right)^i \cdot \deg(v) + 4\sum_{k=0}^{i-1} \left(\frac{2}{3}\right)^{k} \leq \left(\frac{2}{3} \right)^{i} \cdot \deg(v) + 12~.$$

The running time of the $j$'th iteration is $O(2^{j} \cdot T(n, \Delta))$. Therefore, the total running time is $\sum_{j=1}^{i} O(2^{j} \cdot T(n, \Delta)) = O(2^{i} \cdot T(n, \Delta)) $. By setting $i = \log_{3/2} (1/\epsilon)$ and $T(n ,\Delta) = \log^2 n$ from Lemma \ref{lem:sinkless}, we get a $(\epsilon \cdot \deg(v) + 12, (1/\epsilon)^{\log_{3/2} 2})$-path decomposition in $O((1/\epsilon)^{\log_{3/2} 2} \cdot \log^2 n )$ rounds.

Suppose $\mathcal{P}$ is a $(\delta, \lambda)$-path decomposition. If we orient the edges on each path in $\mathcal{P}$ in consistent with the direction of the path, then we must have $|\outdeg(v) - \indeg(v)|\leq \delta(v)$. Therefore, we obtain the following theorem.

\begin{theorem}\label{thm:splitting}
For $\epsilon > 0$, there exists a $O((1/\epsilon)^{1.71} \cdot  \log^2 n )$ rounds deterministic algorithm in the \congest model that computes an orientation such that for each vertex $v$, $|\outdeg(v) - \indeg(v)|\leq \epsilon \deg(v) + 12$.
\end{theorem}

%% file: appendix.tex

\section{Directed Densest Subgraph} \label{sec:directed-DS}
\begin{lemma} [Directed densest subgraph's locality] \label{lem:low-diameter-directed-ds}
 For all directed graphs, there existst  $S,T \subseteq V$ such that $G_{S \cup T}$ has undirected diameter at most $K/\eps \cdot \log n$ for some sufficiently large constant $K$ and furthermore  $d(S,T) \geq (1-\eps)D$.
\end{lemma}
\begin{proof}
Consider an arbitrary directed graph and suppose the maximum density is induced by ${S^*,T^*} \subseteq V$. Let $H$ be the graph such that  $V(H) = S^* \cup T^* $  and $E(H)=\{uv \in E: u \in S^*, v \in T^* \} = E(S^*,T^*)$. Note that the edges in $H$ are the directed edges from a vertex in $S^*$ to a vertex in $T^*$. First, we apply the low-diameter decomposition to  $H$, ignoring edges' directions. Let $C_1,\ldots,C_t$ be the components of the decomposition and let $S_i  = S^* \cap C_i$ and $T_i = T^* \cap C_i$. Note that for every component, both $S_i$ and $T_i$ must be non-empty otherwise it has infinite diameter. Now, suppose for all $i = 1,2,\ldots,t$,
\begin{align*}
\frac{|E(S_i,T_i)|}{\sqrt{|S_i||T_i|}} & < (1-\eps) \frac{|E(S^*,T^*)|}{\sqrt{|S^*||T^*|}} \\
\sum_{i=1}^{t} {|E(S_i,T_i)|} & < (1-\eps)  \frac{|E(S^*,T^*)|}{\sqrt{|S^*| |T^*|}} \sum_{i=1}^{t} \sqrt{|S_i||T_i|} ~.
\end{align*}
Appealing to Cauchy-Schwarz inequality,
\[
 \sum_{i=1}^{t} \sqrt{|S_i||T_i|} \leq \sqrt{\sum_{i=1}^{t} |S_i|}\sqrt{\sum_{i=1}^{t} |T_i|} \leq \sqrt{|S^*| |T^*|}~.
\]
Therefore, 
\[
\sum_{i=1}^{t} {|E(S_i,T_i)|} < (1-\eps) |E(S^*,T^*)|
\]
which is a contradiction since this means that the number of of inter-component edges is more than $\eps |E(S^*,T^*)| = \eps |E(H)|$. Hence, there must be some $(S_i,T_i)$ such that $d(S_i,T_i) \geq (1-\eps)D$. Finally, $G_{S_i \cup T_i} = C_i$ and therefore it has diameter at most $K/\eps \cdot \log n$.
\end{proof}

Unfortunately, if $(S_1,T_1)$ and $(S_2,T_2)$ are disjoint, their union $(S_1 \cup S_2, T_1\cup T_2)$ may have a much smaller density. For example, consider $S_1 = \{ v
\}, T_1 = \{a_1,\ldots,a_x \}, S_2 = \{ b_1,\ldots,b_x\}$, and $T_2 = \{ u\}$. We put directed edges from $v$ to $a_i$ and from $b_i$ to $u$ for all $1 \leq i \leq x$.  Note that $d(S_1,T_1) = d(S_2,T_2) = \sqrt{x}$ whereas $d(S_1 \cup S_2, T_1\cup T_2) = 2 x/(x+1) < 2$. If $(S_1,T_1)$ and $(S_2,T_2)$ are $\Omega(\diameterG)$ far away from each other, it is not possible to have one global output for the entire graph without spending $\Omega(\diameterG)$ time. However, if we are satisfied with local output, then the above also leads to an algorithmic result. In particular, each vertex $u$ will output $(s_u,t_u) \in [n]\times [n]$. Each $S_i = \{ u: s_u = i \} ,T_i =\{ u: t_u = i \}$ for $i > 0$ has density at least $(1-\eps)\tilde{D}$.

The algorithm is analogous to the undirected case with a minor modification. Originally, for each vertex $v$, initialize $s_v$ and $t_v$ to 0. Then, each vertex $v$ collects the subgraph induced by its $r$-neighborhood and computes the densest directed subgraph $(S(v),T(v))$ in that subgraph. If $d(S(v),T(v)) \geq (1-\eps)\tilde{D}$, then $v$ is active. The black vertices are identified similarly as in the undirected case. Each (active) black vertex $v$ will then broadcast $(S(v),T(v))$ to its $r$-neighborhood. If $u \in S(v)$, then $s_u \leftarrow v$. Similarly, if $u \in T(v)$, then $t_u \leftarrow v$. We conclude that there exists a deterministic algorithm for directed version of the dense subgraph problem that runs in $O((\log n) / \epsilon)$ rounds in the \local model.

\section{Omitted Proofs}\label{appendix:omitted-proofs}

\begin{proof}[Proof of Lemma \ref{lem:dual_feasible}]
Define the potential function 
\[
\Phi(t) = \sum_{e\in E} w_e^{(t)}.
\]
Clearly, $\Phi(0) \leq n^2$. We will use the following fact $(1-a)^b \leq e^{-ab} \leq 1-ab+(ab)^2/2$ for $a,b>0$. 

Note that from the algorithm, for any edge $e=uv$ and iteration $t$, we have $\alpha_{eu}^{(t)}, \alpha_{ev}^{(t)} \leq 2$. We have
\begin{align*}
    \Phi(t+1) & = \sum_{e  \in E} w_e^{(t+1)} \\
    & = \sum_{e  \in E} w_e^{(t)} (1-\epsilon)^{\alpha_{eu}^{(t)}+\alpha_{ev}^{(t)}} \\
    & \leq  \sum_{e  \in E} w_e^{(t)} \left( 1 - \epsilon (\alpha_{eu}^{(t)}+\alpha_{ev}^{(t)}) +\frac{\epsilon^2 (\alpha_{eu}^{(t)}+\alpha_{ev}^{(t)})^2}{2} \right) \\
    & \leq  \sum_{e  \in E} w_e^{(t)} \left( 1 - \epsilon (\alpha_{eu}^{(t)}+\alpha_{ev}^{(t)}) +8\epsilon^2  \right) \\
    & = (1+8\epsilon^2) \Phi(t) - \epsilon \sum_{e  \in E}  w_e^{(t)} (\alpha_{eu}^{(t)}+\alpha_{ev}^{(t)})  \\
    & \leq (1+8\epsilon^2) \Phi(t) - \epsilon \sum_{e  \in E}  {w_e^{(t)}} \\
    & = (1-\epsilon+8\epsilon^2)\Phi(t).
    \end{align*}
The last inequality follows from the assumption that $ \sum_{e \in E} w_e^{(t)} (\alpha^{(t)}_{eu} + \alpha^{(t}_{ev} )\geq \sum_{e \in E} w_e^{(t)}$.
Therefore, 
\[
\Phi(T) \leq (1-\epsilon+8\epsilon^2)^{T}    \Phi(0) \leq (1-\epsilon+8\epsilon^2)^{T}    n^2 .
\]
Recall that $\Phi(T) =\sum_{e \in E} w_e^{(T)}=\sum_{e \in E}(1-\epsilon)^{\ell^{(T)}(e)}$. This implies that for each $e \in E$:
\begin{align*}
    (1-\epsilon)^{\ell^{(T)}(e)} & \leq (1-\epsilon+8\epsilon^2)^{T  } n^2  \\
    \ell^{(T)}(e) \ln(1-\epsilon) & \leq T \ln (1-\epsilon+8\epsilon^2) + 2 \ln n \\
    \frac{\ell^{(T)}(e)}{T}  & \geq \frac{\ln (1-\epsilon+8\epsilon^2) }{\ln(1-\epsilon)}+ \frac{2 \ln n}{T \ln(1-\epsilon)}~.
\end{align*}
Let us first use (and verify later) the fact that for $0 < \epsilon < 1/2$:
\begin{align}
    \frac{\ln (1-\epsilon+8\epsilon^2) }{\ln(1-\epsilon)} \geq 1-10\epsilon \label{eq:inequality}~.
\end{align}
Note that $\ln (1-a) \leq -a$ for $0 < a < 1$. From the above, if $T \geq 2/\epsilon^2 \cdot \ln n$, we have
\begin{align*}
    \frac{\ell^{(T)}(e)}{T}  & \geq 1-10\epsilon+ \frac{2 \ln n}{T \ln(1-\epsilon)} \\
    & \geq 1-10\epsilon - \frac{2 \ln n}{\epsilon T } \\
    & \geq 1-10\epsilon -\epsilon = 1-11\epsilon.
\end{align*}
Therefore, 
\begin{align*}
    \frac{\ell^{(T)}(e)}{T} (1+12\epsilon) & \geq (1-11\epsilon)(1+12\epsilon) > 1 \\
    \frac{\sum_{t=1}^T (\alpha_{eu}^{(t)}+\alpha_{ev}^{(t)})}{T} (1+12\epsilon) & > 1 \\
    \alpha_{eu}+\alpha_{ev} & > 1~.
\end{align*}

Also note that in our algorithm, for each $u$ and any $1 \leq t \leq T$, $\sum_{e \ni u} \alpha^{(t)}_{eu} \leq z$. Therefore,
\begin{align*} 
\sum_{e \ni u} \frac{\sum_{t=1}^{T} \alpha^{(t)}_{eu} }{T} \cdot (1+12\epsilon) &\leq (1+12\epsilon) z \\
\sum_{e \ni u} \alpha_{eu} &\leq (1+12\epsilon)z~.
\end{align*}

This implies the solution returned by Algorithm \ref{alg:fractional_dual} is a feasible solution for $\DUAL{(1+12\epsilon)}$.

Finally, we check the inequality (\ref{eq:inequality}). We use the fact that $\ln \left( \frac{1}{1-a}  \right) > a$ for $0 < a < 1$ and $\ln \left( \frac{1}{1-a}  \right) \leq a+a^2$ for $0 < a < 1/2$. Thus, for sufficiently small $\epsilon$,
\begin{align*}
    \frac{\ln (1-\epsilon+8\epsilon^2) }{\ln(1-\epsilon)}  & = \frac{\ln \frac{1}{1-\epsilon+8\epsilon^2}}{\ln \frac{1}{1-\epsilon}} \\
    & \geq \frac{\epsilon-8\epsilon^2}{\epsilon+\epsilon^2} \geq 1-10\epsilon~. \qedhere
\end{align*}

\end{proof}

\begin{proof}[Proof of Lemma \ref{lem:subgraph}]

Recall that to maximize $\sum_{e \in E}w'_e (\alpha_{eu} + \alpha_{ev})$ over $\alpha \in S(z)$,  we do the following. Let $e_1, e_2, \ldots, e_{\deg(u)}$ be the edges incident to $u$ where $\ell(e_1) \leq \ell(e_2) \ldots \leq \ell(e_{\deg(u)})$. We set $\alpha_{e_i u} \leftarrow 2$ for $i = 1, \ldots, \min(\lceil z/2 \rceil - 1, \deg(u))$ and set $\alpha_{e_{\lceil z/2 \rceil u} } \leftarrow z - 2 \cdot (\lceil z/2 \rceil - 1)$ if $\deg(u) \geq \lceil z/2 \rceil$. 

Given a value $\lambda$, let $E_\lambda= \{e \in E: w'_e \geq \lambda \}$ denote the set of edges $e$ such that $w'_{e} \geq \lambda$. Let $\deg_{\lambda}(v)$ be the number of  edges in $E_\lambda$ incident to $v$.  Define 
\[
\deg_{\alpha, \lambda}(v) = \begin{cases} 2 \deg_{\lambda}(v) & \mbox{, if $\deg_{\lambda}(v) < \lceil z/2 \rceil$}\\ z & \mbox{, otherwise.} \end{cases}
\]
Note that $\deg_{\alpha,\lambda}(v) = \sum_{ e \ni v : w'_e \geq \lambda} \alpha_{ev}$. Therefore,

\[\sum_{v \in V}\deg_{\alpha,\lambda}(v) = \sum_{\substack{e \in E: w'_{e} \geq \lambda}} (\alpha_{ev} + \alpha_{eu})~.\] 
Suppose the weights $\{w'_e\}$ are bounded between $a$ and $b$, i.e., $0 \leq a \leq w_e' \leq b$ for all $e$. Then, 
\begin{align*}
\int_{0}^{b} |E_\lambda| d\lambda &= \sum_{e \in E}w'_{e} \\
&> F \cdot \sum_{e \in E}w'_{e} \cdot (\alpha_{eu}+ \alpha_{ev}) \\
&= F \cdot \int_{0}^{b}\sum_{e \in E} [w'_{e} \geq \lambda] \cdot (\alpha_{eu} + \alpha_{ev})  d \lambda \\
&= F \cdot \int_{0}^{b} \sum_{v \in V} \deg_{\alpha,\lambda}(v) d\lambda~.
\end{align*}

Note that if $0 \leq a' < a$,  then $E_{a'} = E_a$ and $\deg_{\alpha, a'}(v) = \deg_{\alpha, a}(v) =\deg(v)$. Therefore, there exists $\lambda \in [a,b]$ such that $|E_\lambda| > F \cdot \sum_{v \in V} \deg_{\alpha,\lambda}(v)$. 
Recall that $V''_\lambda$ is the set of vertices with $\deg_{\lambda}(v) \geq \lceil z / 2 \rceil$ from the defintion in the lemma statement. We have
\begin{align*}
|E(V''_\lambda)| &\geq |E_\lambda| - \sum_{v \in V-V''_\lambda} \deg_{\lambda}(v)\\
&> F \cdot \sum_{v \in V} \deg_{\alpha,\lambda}(v) - \sum_{v \in V-V''_\lambda} \deg_{\lambda}(v) \\
&= F \cdot \left( |V''_\lambda| \cdot z  + \sum_{v \in V - V''_\lambda}2\deg_{\alpha,\lambda}(v) \right) - \sum_{v \in V-V''_\lambda}\deg_{\alpha,\lambda}(v) \\
&\geq F \cdot z \cdot |V''_\lambda|  && \text{since } F > 1/2~.
\end{align*}

Since $|E(V''_\lambda)|  > F \cdot z \cdot |V''_\lambda|  \geq 0$, we deduce that $E(V''_\lambda)$ is non-empty and therefore $V''_\lambda$ is also non-empty. Hence, $d(G[V''_\lambda]) > F \cdot z$.
\end{proof}

\begin{proof}[Proof of Lemma \ref{lem:primal_feasible}]
Consider an iteration $t$ where $\sum_{e \in E}w_e^{(t)} (\alpha_{eu}^{(t)} + \alpha_{ev}^{(t)}) < \sum_{e \in E} w_e^{(t)}$. Let 
$$w'_e = \begin{cases} (1-\epsilon)^{\lceil \ell^{(t)}(e) \rceil}  &, \text{if } \ell^{(t)}(e) \in [\ell_{\min}, \ell_{\max}]\\ 0 &, \mbox{otherwise.}\end{cases}$$

When $w'_e \geq (1-\epsilon)^{\ell_{\max}}$, we have $\ell^{(t)}(e) \leq \ell_{\max}$ and therefore $w'_e \geq (1-\eps)^{\ell^{(t)}(e)+1} = w^{(t)}_e (1-\epsilon)$.  Therefore, $w^{(t)}_e \leq w'_e / (1-\epsilon) + (1-\epsilon)^{\ell_{\max}}$. We have

\begin{align*}
\sum_{e \in E} w'_e (\alpha_{eu}^{(t)} + \alpha_{ev}^{(t)} ) &\leq \sum_{e \in E} w_e^{(t)} (\alpha_{eu}^{(t)} + \alpha_{ev}^{(t)} ) \\& < \sum_{e \in E} w_e^{(t)} \\
& \leq \sum_{e \in E} \left(w'_e / (1-\epsilon) +  (1-\epsilon)^{\ell_{\max }}\right) \\
& = \left( \sum_{e \in E} w'_e / (1-\epsilon) \right)+ m(1-\epsilon)^{\ell_{\max }}~.
\end{align*}

The first inequality holds since $w'_e \leq w_e^{(t)}$. The second inequality follows from our assumption that $ \sum_{e \in E} w_e^{(t)} (\alpha_{eu}^{(t)}+ \alpha_{ev}^{(t)} ) < \sum_{e \in E} w_e^{(t)}$.  Note that 

\begin{align*}
m (1-\epsilon)^{\ell_{\max}} &\leq m (1-\epsilon)^{\ell_{\min}} \cdot (1-\epsilon)^{(1/\epsilon)\log(2m/\epsilon)} \\
&\leq m (1-\epsilon)^{\ell_{\min}} \cdot \exp(-\log(2m/\epsilon)) \\
&\leq \epsilon/2 \cdot (1-\epsilon)^{\ell_{\min}} \\
& \leq \epsilon/2  \cdot \max_{e \in E} w'_e  \leq \epsilon/2 \cdot \sum_{e \in E} w'_e~. 
\end{align*}

Therefore, 
\begin{align*}
\sum_{e \in E} w'_e (\alpha_{eu}^{(t)} + \alpha_{ev}^{(t)} ) &< \left( \sum_{e \in E} w'_e / (1-\epsilon) \right)+ m(1-\epsilon)^{\ell_{\max }} \\
&\leq \left( \sum_{e \in E} w'_e / (1-\epsilon) \right)+ 	\epsilon/2 \cdot \sum_{e \in E} w'_e \\
&\leq \frac{(1+\epsilon/2)}{(1-\epsilon)} \cdot \sum_{e \in E} w'_e \\
\frac{(1-\epsilon)}{(1+\epsilon/2)} \sum_{e \in E}  w'_e (\alpha_{eu}^{(t)} + \alpha_{ev}^{(t)} ) &<  \sum_{e \in E}  w'_e \\
(1-3\epsilon)\cdot  \sum_{e \in E}  w'_e (\alpha_{eu}^{(t)} + \alpha_{ev}^{(t)} ) &<  \sum_{e \in E} w'_e \tag*{$\frac{(1-\epsilon)}{(1+\epsilon/2)}  \geq 1-3\epsilon$ for all $\epsilon \geq 0$~.} 
\end{align*}

Recall that the variables $\{ \alpha^{(t)}_{eu} \}$ maximize $\sum_{e \in E} w^{(t)}_e (\alpha_{eu}^{(t)} + \alpha_{ev}^{(t)} ) $ as follows. Each vertex $u$ distributes the $z$ units of budget to the $\{ \alpha^{(t)}_{eu} \}$ variables of the incident edges with the lowest $\lceil z/2 \rceil$ loads (or 
$\lceil z/2 \rceil$ highest weights). It allocates 2 units of budget to the incident edges with the highest $\lceil z/2 \rceil - 1$ weights and the remaining budget to the remaining edge. We observe that if edges $e_1$ and $e_2$ are incident to $u$, then $w^{(t)}_{e_1} \geq w^{(t)}_{e_2}$ implies $w'_{e_1} \geq w'_{e_2}$. As a result, the variables $\{ \alpha^{(t)}_{eu} \}$ also maximize $\sum_{e \in E} w_e' (\alpha_{eu}^{(t)} + \alpha_{ev}^{(t)} ) $ over $S(z)$.

Now we can apply Lemma \ref{lem:subgraph} with $F = (1-3\epsilon)$ to get a subgraph $G[V''_\lambda]$ whose density is at least $(1-3\epsilon) \cdot z$ for some $\lambda \in \{  (1-\eps)^{\ell_{\max}} ,  (1-\eps)^{\ell_{\max} -1},\ldots,(1-\eps)^{ \ell_{\min} } \}$.  Therefore, Algorithm \ref{alg:mwu-densest-subgraph} will find such a subgraph since $G[V''_\lambda] = G[V'_\ell]$ for some $\ell \in [\ell_{\min}, \ell_{\max}]$.
\end{proof}

\begin{proof}[Proof of Lemma \ref{lem:bits}]
In Algorithm \ref{alg:fractional_dual}, if the parameter $z$ is an integer, then each $\alpha^{(t)}_{eu}$ is either 0,1, or 2. The summation $\sum_{t=1}^{T} \alpha^{(t)}_{eu}$ can be stored using $O(\log T)$ bits.  The average $\sum_{t=1}^{T} \alpha^{(t)}_{eu} / T$ can also be stored using $O(\log T)$ bits since division by $T$ corresponds to a shifting operation (recall that $T$ is a power of 2). 

Recall that in the final solution
\begin{align*}
\alpha_{eu} = (1+2\epsilon)\sum_{t=1}^{T} \alpha^{(t)}_{eu} / T = \underbrace{\sum_{t=1}^{T} \alpha^{(t)}_{eu} / T}_{O(\log T) \text{ bits}}  +\underbrace{2\epsilon \sum_{t=1}^{T} \alpha^{(t)}_{eu} / T}_{O(\log T + \log(1/\eps)) \text{ bits}} ~.\end{align*} 

Multiplying by $2\epsilon$ corresponds to a shifting operation as we assume $\epsilon$ is a negative power of 2. Therefore, it takes $O(\log T + \log(1/\epsilon)) = O(\log \log n  + \log(1/\epsilon))$ bits to represent $\alpha_{eu}$.
\end{proof}